\documentclass[journal,draftcls, onecolumn]{IEEEtran}
\usepackage[dvipdfmx]{graphicx}
\usepackage{amsmath}
\usepackage{amssymb}
\usepackage{amsfonts}
\usepackage{mathrsfs}
\usepackage{theorem}
\usepackage{color}
\DeclareMathAlphabet{\bm}{OML}{cmm}{b}{it}
\usepackage[ruled,vlined]{algorithm2e}
\usepackage{algorithmic}

\theorembodyfont{\rmfamily}
\newtheorem{theorem}{Theorem}
\newtheorem{lemma}{Lemma}

\newtheorem{corollary}{Corollary}
\newtheorem{remark}{Remark}
\newtheorem{proposition}{Proposition}
\newtheorem{example}{Example}

\newcommand{\qed}{\hfill \IEEEQED}

\newcommand{\mc}{-\!\!\!\!\circ\!\!\!\!-}

\newcommand{\bol}[1]{\mathbf{#1}}
\newcommand{\rom}[1]{\mathrm{#1}}

\newcommand{\Pe}{\rom{P}_{\rom{e}}}
\newcommand{\Pc}{\rom{P}_{\rom{c}}}

\newcommand{\argmin}{\mathop{\rm argmin}\limits}

\begin{document}

\title{
Tight Exponential Strong Converses for Lossy Source Coding
with Side-Information and Distributed Function Computation
}

\author{Shun Watanabe}


\maketitle
\begin{abstract}
The exponential strong converse for a coding problem states that, if a coding rate is beyond the
theoretical limit, the correct probability converges to zero exponentially. For the lossy source
coding with side-information, also known as the Wyner-Ziv (WZ) problem, a lower bound on the
strong converse exponent was derived by Oohama. In this paper, we derive the tight strong
converse exponent for the WZ problem; as a special case, we also derive the tight strong
converse exponent for the distributed function computation problem. 
For the converse part, we use the change-of-measure argument developed in the literature and the
soft Markov constraint introduced by Oohama; the matching achievability is proved via the Poisson matching approach recently introduced by
Li and Anantharam. 
Our result is build upon the recently derived tight strong converse exponent for the Wyner-Ahlswede-K\"orner (WAK)
problem; however, compared to the WAK problem, more sophisticated argument
is needed. As an illustration of the necessity of the soft Markov constraint, we present an example
such that the soft Markov constraint is strictly positive.
\end{abstract}

\section{Introduction}

The exponential strong converse for a coding problem states that, if a coding rate is beyond the
theoretical limit, the correct probability converges to zero exponentially. 
Proving such a claim was initiated by Arimoto for the channel coding problem \cite{arimoto:73}; 
later, the strong converse exponent was studied by Dueck and K\"orner in \cite{Due-Kor:79};
see also \cite{Oohama15b} for the equivalence of the two exponents derived in \cite{arimoto:73} and \cite{Due-Kor:79}.
An alternative proof of the strong converse exponent using the change-of-measure argument to be explained later is also available in \cite{Watanabe:24}.

Even though the tight strong converse exponent for point-to-point problems or simple multi-user problems, such as the Slepian-Wolf problem \cite{oohama:94},
have been identified, the strong converse exponent for multi-user problems have been unsolved until recently.
A significant progress was made by Oohama in a series of paper including \cite{Oohama:18,Oohama:19}.
More recently, the tight strong converse exponent of the Wyner-Ahlswede-K\"orner (WAK) problem \cite{wyner:75c,ahlswede:75} was derived in \cite{TakWat:25};
the converse part of \cite{TakWat:25} is based on the change-of-measure argument and the soft Markov constraint to
be explained later in this section. 

In this paper, we consider the lossy source coding with side-information, also known as the 
Wyner-Ziv (WZ) problem \cite{wyner:76}. For this problem, a lower bound on the strong converse exponent was 
derived by Oohama in \cite{Oohama:18}. In this paper, we derive the tight strong converse exponent of the WZ problem;
as a special case, we also derive the tight strong converse exponent of the distributed function computation problem \cite{OrlRoc:01}.
Even though this paper is build upon the result in \cite{TakWat:25}, compared to the WAK problem, the WZ problem is more complicated 
and we need more sophisticated arguments. In fact, it should be mentioned that, even though the strong converse 
of the WAK problem was proved in the mid 70's \cite{ahlswede:76}, the strong converse (not exponential one) of the WZ problem was not solved until recently \cite{Oohama:18}.

As an illustration of our result, let us consider the following simple example of the distributed function computation.
Let $X^n=(X_1,\ldots,X_n)$ and $Y^n=(Y_1,\ldots,Y_n)$ be independent uniform random bit sequences of length $n$.
The encoder sends a message of $Rn$ bits, and the decoder shall compute the AND $Z_j= X_j\wedge Y_j$ of each coordinate,
where $0 < R < 1$. As a time-sharing scheme, if the encoder sends the first $Rn$ bits of $X^n$ and the decoder guesses the remaining
$(1-R)n$ bits of $Z_j$'s, then the correct probability is $(3/4)^{(1-R)n}$.
Even though $X^n$ and $Y^n$ are independent, perhaps surprisingly, our result implies that this time-sharing scheme is not optimal;
see Example \ref{example-AND}.

As mentioned above, our converse uses the change-of-measure argument.
This argument was originally introduced by Gu and Effros in \cite{GuEff:09,GuEff:11}
to prove strong converse for source coding problems where there exists a terminal that observes all
the random variables involved; a particular example is the Gray-Wyner (GW) problem \cite{GraWyn:74}.
In the argument of \cite{GuEff:09,GuEff:11}, we single-letterize information quantities not under the original
source (or channel) but under another modified measure which depends on the code and under which the code is error free.
A difficulty of applying this argument to the so-called distributed coding problems, such as the WAK problem, is that
the characterization of asymptotic limits involve auxiliary random variables and Markov chain constraints. 
By using the idea of soft Markov constraint introduced by Oohama \cite{Oohama:18},
the argument was further developed in \cite{TyaWat20} so that it can be applied to distributed coding problems.
(see also an early attempt in \cite{Watanabe:17}).
More specifically, instead of deriving Markov chain constraints directly, we add 
conditional mutual informations as penalty terms; then, by controlling the multipliers of 
those penalty terms so that the Markov constraints are satisfied at the end, we derive the strong converse.   
More recently, a variation of the change-of-measure argument was further developed by Hamad, Wigger, and Sarkiss in \cite{HamWigSar:23}
so that it can be applied to more involved multi-user networks in a concise manner; rather than adding Markov constraints as penalty terms,
they prove the Markov constraints in an asymptotic limit. 
Although the change-of-measure argument in strong converse originates to \cite{GuEff:09,GuEff:11},
the idea of changing measure via conditioning and
then performing single-letterization can be traced back to the work of Csisz\'ar in the context of the Sanov theorem \cite{csiszar:84}; 
see also a monograph by Yu for thorough exposition of this idea and application
to problems other than strong converse \cite{Yu-monograph}. 

A key observation in \cite{TakWat:25} is the recognition that the soft Markov constraint is a part of the tight exponent
rather than a penalty term that must vanish at the end; namely, we fix the multiplier of the soft Markov constraint to be unity.
In this paper, we follow the same approach. However, compared to the WAK problem, there are new difficulties in the WZ problem.
Particularly, the targeted expression of the exponent is not clear at first. In fact, the strong converse result in \cite{TyaWat20} already
provides some lower bound on the exponent; however, that lower bound is not tight enough, and we cannot 
derive a matching achievability. In the WZ problem, the rate-distortion function can be written as either in 
the mutual information difference form or the conditional mutual information form; the two forms coincide
since the auxiliary random variable satisfies the Markov chain constraint. In \cite{TyaWat20}, the latter form
was employed, and the soft Markov constraint was added on top of that. A crucial observation in this paper is that
the conditional mutual information form can be interpreted as the mutual information difference form plus
the soft Markov constraint; in other words, the conditional mutual information form already contains the
soft Markov constraint. This new observation is the missing piece in \cite{TyaWat20}, which enables us
to derive the tight strong converse exponent; see the paragraph containing \eqref{eq:identity-two-forms} and Remark \ref{remark:comparison-to-TyaWat} for more detail. 

For the achievability, the main difficulty is to identify the role of soft Markov constraints; even though the Markov constraints may be violated
in the expression of the exponent, we can only use coding schemes satisfying the Markov constraints because of 
the nature of distributed coding.
In \cite{TakWat:25}, we derived the soft Markov constraint by evaluating the deviation of Markovity in forward direction
in the sense that the auxiliary random variable shows up on the conditional side. 
In contrast to the WAK problem, the WZ problem involves two soft Markov constraints, and it is difficult to handle the two of them using the same approach as \cite{TakWat:25}.
To overcome this difficulty, we observe that soft Markov constraints can be derived 
by a different principle, evaluation of deviation of Markovity in reverse direction. 
To that end, we use the Poisson matching approach
introduced by Li and Anantharam in \cite{LiAna:21}. 
For more detail, see Remark \ref{remark:deviation-Markovity}.

\paragraph*{Notation and basic facts}

Throughout the paper, random variables (eg.~$X$) and their realizations (eg.~$x$) are denoted by capital and
lower case letters, respectively. Unless otherwise stated, all random variables take values in some finite alphabet
which are denoted by the respective calligraphic letters (eg.~${\cal X}$). The probability distribution of random variable
$X$ is denoted by $P_X$. Similarly, 
$X^n=(X_1,\ldots,X_n)$ and $x^n=(x_1,\ldots,x_n)$ denote, respectively, a random vector and its realization in the $n$th
Cartesian product ${\cal X}^n$. 
For an index $1 \le j \le n$, subsequences $(X_1,\ldots,X_{j-1})$
and $(X_{j+1},\ldots,X_n)$ are abbreviated as $X_j^-$ and $X_j^+$, respectively.  
For a finite set ${\cal S}$, the cardinality of ${\cal S}$ is denoted by $|{\cal S}|$.
For a subset ${\cal T} \subset {\cal S}$, the complement ${\cal S}\backslash {\cal T}$ is denoted by ${\cal T}^c$. 
The indicator function is denoted by $\bol{1}[\cdot]$. 
Information theoretic quantities are denoted in the same manner as \cite{csiszar-korner:11};
e.g.~the entropy and the mutual information are denoted by $H(X)$ and $I(X\wedge Y)$, respectively. 
All information quantities and rates are evaluated with respect to the logarithm of base $2$.

The set of all distributions on ${\cal X}$ is denoted by ${\cal P}({\cal X})$.
The set of all channel from ${\cal X}$ to ${\cal Y}$ is denoted by ${\cal P}({\cal Y}|{\cal X})$.
We will also use the method of types \cite{csiszar-korner:11}. 
The set of all types on ${\cal X}$ is denoted by
${\cal P}_n({\cal X})$. For a given type $P_{\bar{X}} \in {\cal P}_n({\cal X})$, the set of all
sequences with type $P_{\bar{X}}$ is denoted by ${\cal T}_{\bar{X}}$. For a given joint type $P_{\bar{X}\bar{Y}}$
and a sequence $x^n \in {\cal T}_{\bar{X}}$, the set of all sequences whose joint type with $x^n$ is $P_{\bar{X}\bar{Y}}$ 
is denoted by ${\cal T}_{\bar{Y}|\bar{X}}(x^n)$.
For type $P_{\bar{X}}$ and joint type $P_{\bar{X}\bar{Y}}$, we use notations $H(\bar{X})$ and $I(\bar{X}\wedge \bar{Y})$,
where the random variables are distributed according to those type and joint type.



\section{Problem Formulation and Results} \label{section:problem-formulation}

In this section, we describe the problem formulation of the lossy source coding with side-information,
also known as the Wyner-Ziv problem. 
Let us consider a correlated source $(X,Y)$ taking values in ${\cal X}\times {\cal Y}$ and having joint distribution $P_{XY}$.
We consider a block coding of length $n$. A coding system consists of an encoder
\begin{align*}
\varphi_n:{\cal X}^n \to {\cal M}^{(n)},
\end{align*}
and a decoder 
\begin{align*}
\psi_n: {\cal M}^{(n)} \times {\cal Y}^n \to {\cal Z}^n,
\end{align*}
where ${\cal Z}$ is a reproduction alphabet. 
In the following, we omit the blocklength $n$ when it is obvious from the context. 
Let $d:{\cal X}\times {\cal Y}\times  {\cal Z}\to \mathbb{R}_+$ be a distortion measure,\footnote{Even though it is more common to
 consider a distortion measure $d:{\cal X}\times {\cal Z}\to \mathbb{R}_+$ in the Wyner-Ziv problem, we consider this generalized 
 distortion measure so that the distributed function computation problem can be treated as a special case of the Wyner-Ziv problem.} 
where $\mathbb{R}_+$ is the set of all non-negative real numbers.
For sequences $x^n=(x_1,\ldots,x_n) \in {\cal X}^n$, $y^n=(y_1,\ldots,y_n)$, and $z^n=(z_1,\ldots,z_n)\in {\cal Z}^n$,
we consider the additive distortion given by $d_n(x^n,y^n,z^n) = \sum_{j=1}^n d(x_j,y_j,z_j)$.
For a given code $\Phi_n=(\varphi,\psi)$ and a threshold $D\ge 0$ of distortion, we are interested in the excess distortion probability
\begin{align*}
\Pe(\Phi_n) := \Pr\bigg( \frac{1}{n}d_n(X^n, Y^n, \psi(\varphi(X^n),Y^n)) > D \bigg)
\end{align*}
and non-excess distortion probability
\begin{align*}
\Pc(\Phi_n) := 1 - \Pe(\Phi_n),
\end{align*}
where $(X^n,Y^n)$ are independently identically distributed (i.i.d.) random variables with distribution $P_{XY}$.

For a given distortion level $D\ge 0$, a rate $R$ is defined to be achievable if there exists a sequence of code $\{\Phi_n\}_{n=1}^\infty$
such that 
\begin{align*}
\limsup_{n\to\infty} \frac{1}{n} \log|{\cal M}^{(n)}| \le R
\end{align*} 
and 
\begin{align*}
\lim_{n\to\infty} \Pe(\Phi_n)  = 0.
\end{align*}
Then, the rate-distortion function $R_{\mathtt{WZ}}(D)$ is defined as the infimum of achievable rates.\footnote{It is common
to consider the expected distortion constraint rather than the excess distortion constraint; however, the rate-distortion functions
for the two constraints are equivalent for finite alphabets.}

The rate-distortion function $R_{\mathtt{WZ}}(D)$ was characterized in \cite{wyner:76}.
Let ${\cal U}$ be an auxiliary alphabet such that $|{\cal U}| \le |{\cal X}|+2$, and let 
\begin{align}
R^\star_{\mathtt{WZ}}(D) &:= \min\big[ I(U\wedge X) - I(U\wedge Y)\big] \label{eq:WZ-RD-difference-form} \\
&= \min I(U \wedge X|Y), \label{eq:WZ-RD-conditional-form}
\end{align}
where the minimization is taken over auxiliary random variable $U$ on ${\cal U}$ satisfying the Markov chain 
constraint $U\mc X \mc Y$ such that there exists a reproduction random variable $Z$ on ${\cal Z}$ satisfying the Markov chain
constraint $Z \mc (U,Y) \mc X$ and the distortion constraint $\mathbb{E}[d(X,Y,Z)] \le D$. In fact, in order to characterize the rate-distortion 
function, it suffices to consider $Z$ that is a deterministic function of $(U,Y)$; however, since we consider stochastic function $Z$ later 
in the strong converse exponent, we consider stochastic functions here as well for consistency. Note also that 
the expressions \eqref{eq:WZ-RD-difference-form} and \eqref{eq:WZ-RD-conditional-form} coincide because of $U\mc X \mc Y$.
As we will see below, the two expressions play different roles in the analysis of the strong converse exponent.

\begin{proposition}
For a given distortion level $D\ge 0$, it holds that
\begin{align*}
R_{\mathtt{WZ}}(D) = R^\star_{\mathtt{WZ}}(D).
\end{align*}
\end{proposition}

In this paper, we are interested in the exponential behavior of
the non-excess probability when the rate $R$ is strictly smaller than $R_{\mathtt{WZ}}(D)$.
To that end, let 
\begin{align} \label{eq:FnRD}
F^{(n)}(R,D) := \min_{\Phi_n} \frac{1}{n} \log \frac{1}{\Pc(\Phi_n)} 
\end{align}
and
\begin{align*}
F(R,D) := \lim_{n\to\infty} F^{(n)}(R,D),
\end{align*}
where the minimum \eqref{eq:FnRD} is taken over all codes such that $\frac{1}{n}\log|{\cal M}^{(n)}|\le R$.
In fact, it has been shown in \cite{Oohama:18} that $F^{(n)}(R,D)$ is non-increasing function, and thus the limit in the definition of $F(R,D)$ exists.

For auxiliary alphabet ${\cal U}$ such that $|{\cal U}|\le |{\cal X}||{\cal Y}||{\cal Z}|+1$, let 
\begin{align*}
F^\star(R,D) := \min_{P_{\tilde{U}\tilde{X}\tilde{Y}\tilde{Z}}:\atop \mathbb{E}[d(\tilde{X},\tilde{Y},\tilde{Z})]\le D}\big[ D(P_{\tilde{U}\tilde{X}\tilde{Y}\tilde{Z}} \| P_{\tilde{Z}|\tilde{U}\tilde{Y}}P_{\tilde{U}|\tilde{X}}P_{XY})
 + | I(\tilde{U} \wedge \tilde{X}) - I(\tilde{U} \wedge \tilde{Y}) -R |^+\big],
\end{align*}
where the minimization is taken over joint distributions $P_{\tilde{U}\tilde{X}\tilde{Y}\tilde{Z}}$ satisfying the distortion constraint
$\mathbb{E}[d(\tilde{X},\tilde{Y},\tilde{Z})]\le D$, and
$|a|^+ := \max[a,0]$. 
It should be emphasized that we do not impose Markov chain constraints
$\tilde{U}\mc \tilde{X}\mc \tilde{Y}$ nor $\tilde{Z}\mc (\tilde{U},\tilde{Y})\mc \tilde{X}$. Instead, since
\begin{align}
D(P_{\tilde{U}\tilde{X}\tilde{Y}\tilde{Z}} \| P_{\tilde{Z}|\tilde{U}\tilde{Y}}P_{\tilde{U}|\tilde{X}}P_{XY})
&= D(P_{\tilde{X}\tilde{Y}}\|P_{XY}) + D(P_{\tilde{U}|\tilde{X}\tilde{Y}}\|P_{\tilde{U}|\tilde{X}}|P_{\tilde{X}\tilde{Y}}) 
+ D(P_{\tilde{Z}|\tilde{U}\tilde{X}\tilde{Y}}\|P_{\tilde{Z}|\tilde{U}\tilde{Y}}|P_{\tilde{U}\tilde{X}\tilde{Y}}) \nonumber \\
&= D(P_{\tilde{X}\tilde{Y}}\|P_{XY}) + I(\tilde{U} \wedge \tilde{Y}|\tilde{X}) + I(\tilde{Z} \wedge \tilde{X}|\tilde{U},\tilde{Y}), \label{eq:WZ-soft-Markov}
\end{align}
the two Markov constraints are implicitly incorporated into $F^\star(R,D)$
as soft Markov constraints. Since the Markov constraints are satisfied if the soft Markov constraints are zero,
we can verify that $F^\star(R,D)=0$ implies $R \ge R^\star_{\mathtt{WZ}}(D)$, i.e., 
$R < R^\star_{\mathtt{WZ}}(D)$ implies $F^\star(R,D)>0$.

It is also worth mentioning that the mutual information difference form
$I(\tilde{U}\wedge \tilde{X})-I(\tilde{U}\wedge \tilde{Y})$ of the rate-distortion function (cf.~\eqref{eq:WZ-RD-difference-form}) appears
in $F^\star(R,D)$. Since the identity
\begin{align} \label{eq:identity-two-forms}
I(\tilde{U}\wedge \tilde{X}|\tilde{Y}) = I(\tilde{U}\wedge \tilde{Y}|\tilde{X}) + I(\tilde{U}\wedge \tilde{X}) - I(\tilde{U}\wedge \tilde{Y})
\end{align}
holds, the conditional mutual information form $I(\tilde{U}\wedge \tilde{X}|\tilde{Y})$ of the rate-distortion function (cf.~\eqref{eq:WZ-RD-conditional-form})
can be interpreted as the difference form plus the soft Markov constraint. 
Even though the identity \eqref{eq:identity-two-forms} is just a chain rule, the above interpretation has never been made in the literature.
In fact, this interpretation
plays a crucial role in the converse proof of our main result; see also Remark \ref{remark:comparison-to-TyaWat}. 

Now, we are ready to state our main result. 
\begin{theorem}\label{theorem:WZ}
For every $R\ge 0$ and $D\ge 0$, we have
\begin{align*}
F(R,D) = F^\star(R,D).
\end{align*}
\end{theorem}

\begin{remark}
Even though it is not immediately clear if the characterization $F^\star(R,D)$ is a convex function of $(R,D)$,
we can verify that it is indeed a convex function via the operationally defined quantity $F(R,D)$.
In fact, since the non-excess probability of a concatenation of two codes with non-overlapping blocks is
the product of the non-excess probability of each code, the convexity of the exponent can be proved by the time-sharing 
argument.\footnote{Note that the same argument cannot be used to prove 
the convexity of error exponent.} Fore more detail of this argument, see \cite[Property 1]{Oohama15b}.
\end{remark}

\begin{remark}\label{remark:optimizer}
In the optimization of $F^\star(R,D)$, the term $I(\tilde{U}\wedge \tilde{X})-I(\tilde{U}\wedge \tilde{Y})$
may be negative for some choice of $P_{\tilde{U}\tilde{X}\tilde{Y}}$. However, we can show that such a
$P_{\tilde{U}\tilde{X}\tilde{Y}\tilde{Z}}$ is not an optimizer of $F^\star(R,D)$. In fact, if such a $P_{\tilde{U}\tilde{X}\tilde{Y}\tilde{Z}}$ is an optimizer,
then we have
\begin{align*}
F^\star(R,D) &= D(P_{\tilde{X}\tilde{Y}}\|P_{XY}) + I(\tilde{U}\wedge \tilde{Y}|\tilde{X}) + I(\tilde{Z}\wedge \tilde{X}|\tilde{U},\tilde{Y}) \\
&>   D(P_{\tilde{X}\tilde{Y}}\|P_{XY}) + I(\tilde{U}\wedge \tilde{Y}|\tilde{X}) + I(\tilde{Z}\wedge \tilde{X}|\tilde{U},\tilde{Y}) 
 + I(\tilde{U}\wedge \tilde{X}) - I(\tilde{U}\wedge \tilde{Y}) \\
&= D(P_{\tilde{X}\tilde{Y}}\|P_{XY}) + I(\tilde{U}\wedge \tilde{X}|\tilde{Y}) + I(\tilde{Z}\wedge \tilde{X}|\tilde{U},\tilde{Y}) \\
&\ge F(R,D) \\
&=F^\star(R,D),
\end{align*}
which is a contradiction. In the second inequality above, we used a naive achievability bound \eqref{eq:naive-achievability} in Remark \ref{ramark:naive-achievability};
we do not know if it can be proved just by an analysis of information quantities.
\end{remark}

\begin{corollary}[Distributed function computation] \label{corollary:dfc}
When the distortion is given by
\begin{align*}
d(x,y,z) = \bol{1}[z=f(x,y)]
\end{align*}
for a function $f:{\cal X}\times {\cal Y}\to {\cal Z}$ and $D=0$, then the problem reduces to the distributed 
function computation problem. In this case, the expression of the exponent can be written as\footnote{Rigorously, 
the case $D=0$ requires a separate proof for the achievability; see Remark \ref{remark:proof-dfc}.} 
\begin{align}
F^\star(R,0) = \min_{P_{\tilde{U}\tilde{X}\tilde{Y}\tilde{Z}} : \atop \tilde{Z}=f(\tilde{X},\tilde{Y})}
\big[
D(P_{\tilde{U}\tilde{X}\tilde{Y}} \| P_{\tilde{U}|\tilde{X}}P_{XY}) + H(\tilde{Z}|\tilde{U},\tilde{Y})
 + | I(\tilde{U} \wedge \tilde{X}) - I(\tilde{U} \wedge \tilde{Y}) -R |^+
\big]. \label{eq:expression-function-computation}
\end{align}
\end{corollary}
\begin{proof}
Since $\mathbb{E}[d(\tilde{X},\tilde{Y},\tilde{Z})]=0$ implies $H(\tilde{Z}|\tilde{X},\tilde{Y})=0$, we have
\begin{align*}
I(\tilde{Z}\wedge \tilde{X}|\tilde{U},\tilde{Y}) = H(\tilde{Z}|\tilde{U},\tilde{Y}) - H(\tilde{Z}|\tilde{U},\tilde{X},\tilde{Y}) = H(\tilde{Z}|\tilde{U},\tilde{Y}).
\end{align*}
Then, the expression \eqref{eq:expression-function-computation} follows from Theorem \ref{theorem:WZ} and \eqref{eq:WZ-soft-Markov}.
\end{proof}

\begin{corollary}[Slepian-Wolf coding]
When the distortion is given by
\begin{align*}
d(x,y,z) = \bol{1}[z=x]
\end{align*}
and $D=0$, then the problem reduces to the lossless source coding with full side-information, also known as 
the Slepian-Wolf coding. In this case, $F^\star(R,0)$ can be written as
\begin{align*}
F_{\mathtt{SW}}^\star(R) := \min_{P_{\tilde{X}\tilde{Y}}} \big[ D(P_{\tilde{X}\tilde{Y}}\|P_{XY}) + | H(\tilde{X}|\tilde{Y})-R|^+ \big].
\end{align*}
\end{corollary}
\begin{proof}
We can verify $F^\star(R,0)=F_{\mathtt{SW}}^\star(R)$ as follows. First, by setting $\tilde{U}=\tilde{Z}=\tilde{X}$, the expression 
$F^\star(R,0)$ becomes $F_{\mathtt{SW}}^\star(R)$, i.e., $F_{\mathtt{SW}}^\star(R) \ge F^\star(R,0)$.
To verify the opposite inequality, for any $P_{\tilde{U}\tilde{X}\tilde{Y}\tilde{Z}}$, 
we have
 \begin{align*}
 \lefteqn{
D(P_{\tilde{X}\tilde{Y}}\|P_{XY}) + I(\tilde{U} \wedge \tilde{Y}|\tilde{X}) + I(\tilde{Z} \wedge \tilde{X}|\tilde{U},\tilde{Y}) + | I(\tilde{U} \wedge \tilde{X}) - I(\tilde{U}\wedge \tilde{Y})-R|^+
} \\
&\ge D(P_{\tilde{X}\tilde{Y}}\|P_{XY}) + I(\tilde{U}\wedge \tilde{Y}|\tilde{X}) + I(\tilde{Z}\wedge \tilde{X}|\tilde{U},\tilde{Y}) +  I(\tilde{U}\wedge \tilde{X}) - I(\tilde{U}\wedge \tilde{Y})-R \\
&= D(P_{\tilde{X}\tilde{Y}}\|P_{XY})+ I(\tilde{Z}\wedge \tilde{X}|\tilde{U},\tilde{Y}) + I(\tilde{U}\wedge \tilde{X}|\tilde{Y}) - R \\
&= D(P_{\tilde{X}\tilde{Y}}\|P_{XY})+ H(\tilde{X}|\tilde{U},\tilde{Y}) - H(\tilde{X}|\tilde{U},\tilde{Y}, \tilde{Z}) + H(\tilde{X}|\tilde{Y}) - H(\tilde{X}|\tilde{U},\tilde{Y}) - R \\
&= D(P_{\tilde{X}\tilde{Y}}\|P_{XY})+ H(\tilde{X}|\tilde{Y}) - R,
\end{align*} 
where we used $H(\tilde{X}|\tilde{U},\tilde{Y}, \tilde{Z})=0$, which follows from $\mathbb{E}[d(\tilde{X},\tilde{Y},\tilde{Z})]=0$, in the last equality. On the other hand, we also have
\begin{align*}
D(P_{\tilde{X}\tilde{Y}}\|P_{XY}) + I(\tilde{U} \wedge \tilde{Y}|\tilde{X}) + I(\tilde{Z} \wedge \tilde{X}|\tilde{U},\tilde{Y}) + | I(\tilde{U} \wedge \tilde{X}) - I(\tilde{U}\wedge \tilde{Y})-R|^+
\ge D(P_{\tilde{X}\tilde{Y}}\|P_{XY}).
\end{align*}
Thus, we have $F^\star(R,0) \ge F_{\mathtt{SW}}^\star(R)$.
\end{proof}

\begin{figure}[tb]
\centering{
\includegraphics[width=0.5\linewidth]{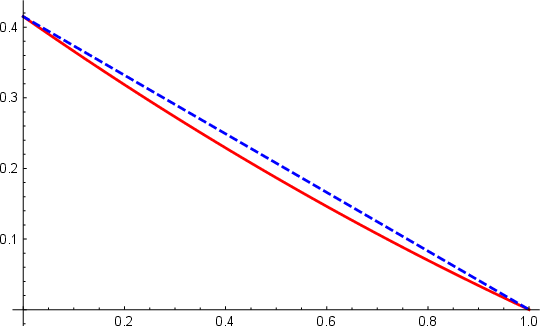}
\caption{A comparison of \eqref{eq:time-sharing-exponent-AND} (blue dashed line) and 
\eqref{eq:coded-exponent-AND} (red solid line).
The horizontal axis is $0 \le R \le 1$ and the vertical axis is the correct probability exponent.}
\label{Fig:AND-comparison}
}
\end{figure}

\begin{example}[AND function for independent sources] \label{example-AND}
Let ${\cal X}={\cal Y}={\cal Z}=\{0,1\}$, and let $f(x,y)=x \wedge y$ be the AND function. 
Consider the uniform distribution $P_{XY}(x,y)=\frac{1}{4}$ for every $(x,y)\in {\cal X}\times {\cal Y}$.
The optimal rate needed to reproduce $Z=f(X,Y)$ is known to be $H(X)=1$, i.e, the encoder need to send the entire source;
eg.~see \cite{HanKob87}. When the rate is $0 < R < 1$, the encoder can send the entire source for the first $Rn$ bit;
the decoder guesses the value of function as $0$ for the remaining $(1-R)n$ bits; then the correct probability of this time-sharing scheme $\Phi_n^{\mathtt{TS}}$
is
\begin{align*}
\Pc(\Phi_n^{\mathtt{TS}}) = \bigg(\frac{3}{4}\bigg)^{(1-R)n},
\end{align*}
i.e., the correct probability exponent is
\begin{align}
\frac{1}{n}\log \frac{1}{\Pc(\Phi_n^{\mathtt{TS}})} = (1-R)(2 - \log 3). \label{eq:time-sharing-exponent-AND}
\end{align}
Since the sources $(X,Y)$ are independent, one may wonder there is no scheme that is superior to \eqref{eq:time-sharing-exponent-AND}.
Perhaps surprisingly, our optimal exponent satisfies 
\begin{align} \label{eq:coded-exponent-AND}
F^\star(R,0) \le 2 - h\bigg(\frac{2+R}{6}\bigg) - \frac{2+R}{3}.
\end{align}
Note that the strict concavity of $h(\cdot)$ implies that the right-hand side of \eqref{eq:coded-exponent-AND} is strictly convex.
Thus, since \eqref{eq:time-sharing-exponent-AND} and \eqref{eq:coded-exponent-AND} coincide at $R=0$ and $R=1$, 
\eqref{eq:coded-exponent-AND} is strictly smaller than \eqref{eq:time-sharing-exponent-AND} for $0 < R < 1$; see Fig.~\ref{Fig:AND-comparison}.

To verify \eqref{eq:coded-exponent-AND}, for ${\cal U}=\{0,1,2\}$, we set the joint distribution $P_{\tilde{U}\tilde{X}\tilde{Y}}$ 
as follows: let $P_{\tilde{U}}(0)=P_{\tilde{U}}(1)=\frac{R}{2}$, and $P_{\tilde{U}}(2)=(1-R)$; then, let
\begin{align*}
P_{\tilde{X}\tilde{Y}|\tilde{U}}(\cdot,\cdot|0) &= \left[
\begin{array}{cc}
\frac{1}{2} & \frac{1}{2} \\
0 & 0
\end{array}
\right], \\
P_{\tilde{X}\tilde{Y}|\tilde{U}}(\cdot,\cdot|1) &= \left[
\begin{array}{cc}
0 & 0 \\
\frac{1}{2} & \frac{1}{2} 
\end{array}
\right], \\
P_{\tilde{X}\tilde{Y}|\tilde{U}}(\cdot,\cdot|2) &= \left[
\begin{array}{cc}
\frac{1}{3} & \frac{1}{3} \\
\frac{1}{3} & 0 
\end{array}
\right].
\end{align*}
Note that
\begin{align*}
P_{\tilde{X}\tilde{Y}} = R \cdot
\left[
\begin{array}{cc}
\frac{1}{4} & \frac{1}{4} \\
\frac{1}{4} & \frac{1}{4}
\end{array}
\right]
+ (1-R) \cdot 
\left[
\begin{array}{cc}
\frac{1}{3} & \frac{1}{3} \\
\frac{1}{3} & 0
\end{array}
\right]
= \left[
\begin{array}{cc}
\frac{4-R}{12} & \frac{4-R}{12} \\
\frac{4-R}{12} & \frac{R}{4}
\end{array}
\right].
\end{align*}
For this choice of joint distribution, we can compute
\begin{align}
I(\tilde{U}\wedge \tilde{Y}|\tilde{X}) &= \frac{4-R}{6} + \frac{2+R}{6} h\bigg( \frac{4-R}{4+2R} \bigg) - \frac{2+R}{3}, \label{eq:Markov-constraint-AND} \\
H(\tilde{Z}|\tilde{U},\tilde{Y}) &= 0, \\
I(\tilde{U}\wedge \tilde{X}) - I(\tilde{U}\wedge \tilde{Y}) &= R.
\end{align}
Thus, by substituting these quantities into \eqref{eq:expression-function-computation}, we have
\begin{align*}
 2 - H\bigg(\frac{4-R}{12},\frac{4-R}{12}, \frac{4-R}{12}, \frac{R}{4} \bigg) + \frac{4-R}{6} + \frac{2+R}{6} h\bigg( \frac{4-R}{4+2R} \bigg) - \frac{2+R}{3} 
= 2 - h\bigg(\frac{2+R}{6}\bigg) - \frac{2+R}{3}.
\end{align*}

By noting the inequality $h(\theta) > 1-2 \theta$ for $\frac{1}{2}< \theta < 1$, we can verify that
the constraint \eqref{eq:Markov-constraint-AND} satisfies $I(\tilde{U}\wedge \tilde{Y}|\tilde{X}) > 0$ for $0 < R < 1$.
In fact, we can show that the optimizer $P_{\tilde{U}\tilde{X}\tilde{Y}}$ of \eqref{eq:expression-function-computation} 
for this example must satisfy either $I(\tilde{U}\wedge \tilde{Y}|\tilde{X})>0$ or $H(\tilde{Z}|\tilde{U},\tilde{Y})>0$.
To verify this, first note that the optimizer $P_{\tilde{U}\tilde{X}\tilde{Y}}$ must be such that $\mathtt{supp}(P_{\tilde{X}\tilde{Y}})={\cal X}\times {\cal Y}$;
otherwise, $|\mathtt{supp}(P_{\tilde{X}\tilde{Y}})| \le 3$ and $F^\star(R,0)$ cannot be smaller than
\begin{align*}
D(P_{\tilde{X}\tilde{Y}}\| P_{XY}) = 2 - H(\tilde{X},\tilde{Y}) \ge 2 - \log 3.
\end{align*}
When $\mathtt{supp}(P_{\tilde{X}\tilde{Y}})={\cal X}\times {\cal Y}$, we can show that $I(\tilde{U} \wedge \tilde{Y}|\tilde{X})=0$
and $H(\tilde{Z}|\tilde{U},\tilde{Y})=0$ imply $H(\tilde{X}|\tilde{U})=0$ in the same manner as \cite[Lemma 1]{HanKob87}.
In fact, if $H(\tilde{X}|\tilde{U})>0$, the there exists some $u\in {\cal U}$ such that $P_{\tilde{U}\tilde{X}}(u,0)>0$
and $P_{\tilde{U}\tilde{X}}(u,1)>0$. Since $P_{\tilde{X}\tilde{Y}}$ has full support and $\tilde{U} \mc \tilde{X} \mc \tilde{Y}$, we have
\begin{align*}
P_{\tilde{U}\tilde{Z}\tilde{Y}}(u,0,1) &= P_{\tilde{U}\tilde{X}\tilde{Y}}(u,0,1) > 0, \\
P_{\tilde{U}\tilde{Z}\tilde{Y}}(u,1,1) &= P_{\tilde{U}\tilde{X}\tilde{Y}}(u,1,1) > 0.
\end{align*}
Thus, we have $H(\tilde{Z}|\tilde{U}=u,\tilde{Y}=1)>0$, which contradict $H(\tilde{Z}|\tilde{U},\tilde{Y})=0$.
Consequently, if the optimizer $P_{\tilde{U}\tilde{X}\tilde{Y}}$ is such that $I(\tilde{U} \wedge \tilde{Y}|\tilde{X})=0$ and $H(\tilde{Z}|\tilde{U},\tilde{Y})=0$,
then $F^\star(R,0)$ cannot be smaller than
\begin{align*}
D(P_{\tilde{X}\tilde{Y}}\| P_{XY}) + | H(\tilde{X}|\tilde{Y}) - R|^+ 
&\ge D(P_{\tilde{X}\tilde{Y}}\| P_{XY}) + H(\tilde{X}|\tilde{Y}) - R \\
&= 2 - H(\tilde{Y}) - R \\
&\ge 1- R.
\end{align*}
This guarantees that at least one of the soft Markov constraints must be strictly positive to attain the optimal 
value of $F^\star(R,0)$.
\end{example}

\section{Proof of Converse}

In this section, we prove the converse part of Theorem \ref{theorem:WZ}.
Our proof is based on the change-of-measure argument.
We follow, in part, similar steps as the strong converse proof in \cite{TyaWat20};
however, to derive the tight exponent, we need more judicious argument; see Remark \ref{remark:comparison-to-TyaWat}.

For a given code $\Phi_n=(\varphi,\psi)$, let
\begin{align*}
{\cal C} = \big\{ (x^n,y^n): d_n(x^n, y^n,\psi(\varphi(x^n),y^n)) \le n D \big\}
\end{align*}
be the set of all sequences that do not cause the excess distortion, and let $P_{\tilde{X}^n\tilde{Y}^n}$
be the probability distribution obtained by conditioning $P_{X^n Y^n}$ on ${\cal C}$, i.e.,
\begin{align*}
P_{\tilde{X}^n\tilde{Y}^n}(x^n,y^n) := \frac{P_{X^nY^n}(x^n,y^n)\bol{1}[(x^n,y^n) \in {\cal C}]}{P_{X^nY^n}({\cal C})}.
\end{align*}
Then, we can verify
\begin{align} \label{eq:WZ-changed-measure-divergence}
D(P_{\tilde{X}^n\tilde{Y}^n} \| P_{X^nY^n}) = \log \frac{1}{\Pc(\Phi_n)}.
\end{align}

For random variables $(\tilde{X}^n,\tilde{Y}^n)$ distributed according to $P_{\tilde{X}^n\tilde{Y}^n}$, let $\tilde{M}= \varphi(\tilde{X}^n)$
and $\tilde{Z}^n = \psi(\tilde{M}, \tilde{Y}^n)$. Since $\tilde{M} \mc \tilde{X}^n \mc \tilde{Y}^n$ and 
\begin{align*}
0 = H(\tilde{Z}^n|\tilde{M},\tilde{Y}^n) = I(\tilde{Z}^n \wedge \tilde{X}^n|\tilde{M},\tilde{Y}^n)
\end{align*}
hold,
the identity \eqref{eq:WZ-changed-measure-divergence} implies
\begin{align} \label{eq:WZ-multiletter-bound-1}
\log \frac{1}{\Pc(\Phi_n)} = D(P_{\tilde{X}^n\tilde{Y}^n} \| P_{X^nY^n}) + I(\tilde{M}\wedge \tilde{Y}^n | \tilde{X}^n) + I(\tilde{Z}^n \wedge \tilde{X}^n|\tilde{M},\tilde{Y}^n).
\end{align}
On the other hand, since
\begin{align*}
nR \ge H(\tilde{M}) \ge I(\tilde{M}\wedge \tilde{X}^n|\tilde{Y}^n),
\end{align*}
the identity \eqref{eq:WZ-changed-measure-divergence} also implies
\begin{align} \label{eq:WZ-multiletter-bound-2}
\log \frac{1}{\Pc(\Phi_n)} \ge D(P_{\tilde{X}^n\tilde{Y}^n} \| P_{X^nY^n}) + I(\tilde{M}\wedge \tilde{X}^n|\tilde{Y}^n) + I(\tilde{Z}^n \wedge \tilde{X}^n|\tilde{M},\tilde{Y}^n) - nR.
\end{align}
Since the excess distortion never occurs on the set ${\cal C}$, we have
\begin{align} \label{eq:WZ-distortion-bound}
D \ge \mathbb{E}\bigg[ \sum_{j=1}^n \frac{1}{n} d(\tilde{X}_j, \tilde{Y}_j,\tilde{Z}_j) \bigg] 
= \mathbb{E}[ d(\tilde{X}_J, \tilde{Y}_J, \tilde{Z}_J)],
\end{align}
where $J$ is the random variable uniformly distributed on the index set $\{1,\ldots,n\}$.
The rest of the task is to single-letterize the right hand sides of \eqref{eq:WZ-multiletter-bound-1} and \eqref{eq:WZ-multiletter-bound-2}.

In the single-letterization, we use the following inequality from \cite[Proposition 1]{TyaWat20}:
\begin{align} \label{eq:single-letterize-divergence-entropy}
D(P_{\tilde{X}^n\tilde{Y}} \| P_{X^n Y^n}) + H(\tilde{X}^n|\tilde{Y}^n) \ge n \big( D(P_{\tilde{X}_J \tilde{Y}_J}\|P_{XY}) + H(\tilde{X}_J|\tilde{Y}_J) \big).
\end{align}
The following simple identity for any triplets $(A,B,C)$ of random variables also plays a pivotal role (cf.~\eqref{eq:identity-two-forms}): 
\begin{align} \label{eq:three-variable-identity}
I(A \wedge B|C) = I(A \wedge C|B) + I(A \wedge B) - I(A \wedge C).
\end{align}

First, we single-letterize \eqref{eq:WZ-multiletter-bound-2} as follows (the term $nR$ is omitted):
\begin{align}
\lefteqn{ D(P_{\tilde{X}^n \tilde{Y}^n} \| P_{X^n Y^n}) + I(\tilde{M} \wedge \tilde{X}^n | \tilde{Y}^n) + I(\tilde{Z}^n \wedge \tilde{X}^n|\tilde{M},\tilde{Y}^n) } \\
&= D(P_{\tilde{X}^n \tilde{Y}^n} \| P_{X^n Y^n}) + H(\tilde{X}^n | \tilde{Y}^n) - H(\tilde{X}^n | \tilde{M}, \tilde{Y}^n,\tilde{Z}^n) \label{eq:WZ-proof-1} \\
&\ge n D( P_{\tilde{X}_J \tilde{Y}_J} \| P_{XY}) + n H(\tilde{X}_J | \tilde{Y}_J) - \sum_{j=1}^n H(\tilde{X}_j | \tilde{M}, \tilde{X}_j^-, \tilde{Y}^n,\tilde{Z}^n)  \label{eq:WZ-proof-2} \\
&\ge n D( P_{\tilde{X}_J \tilde{Y}_J} \| P_{XY}) + n H(\tilde{X}_J | \tilde{Y}_J) - \sum_{j=1}^n H(\tilde{X}_j | \tilde{M}, \tilde{X}_j^-, \tilde{Y}_j, \tilde{Y}_j^+, \tilde{Z}_j) \label{eq:WZ-proof-3} \\
&= n \big[  D( P_{\tilde{X}_J \tilde{Y}_J} \| P_{XY})  +  H(\tilde{X}_J | \tilde{Y}_J) - H(\tilde{X}_J | \tilde{U}_J, J, \tilde{Y}_J, \tilde{Z}_J) \big] \label{eq:WZ-proof-4} \\
&= n \big[ D( P_{\tilde{X}_J \tilde{Y}_J} \| P_{XY}) + I(\tilde{U}_J, J \wedge \tilde{X}_J | \tilde{Y}_J) + I(\tilde{Z}_J \wedge \tilde{X}_J |\tilde{U}_J,J, \tilde{Y}_J) \big]   \\
&= n \big[ D( P_{\tilde{X}_J \tilde{Y}_J} \| P_{XY}) + I(\tilde{U}_J, J \wedge \tilde{Y}_J | \tilde{X}_J) + I(\tilde{Z}_J \wedge \tilde{X}_J |\tilde{U}_J,J, \tilde{Y}_J)
 + I( \tilde{U}_J, J \wedge \tilde{X}_J) - I(\tilde{U}_J,J \wedge \tilde{Y}_J) \big], \label{eq:WZ-proof-5}
\end{align}
where we set $\tilde{U}_j = (\tilde{M},\tilde{X}_j^-,\tilde{Y}_j^+)$;  
\eqref{eq:WZ-proof-2} follows from \eqref{eq:single-letterize-divergence-entropy} and the chain rule;
\eqref{eq:WZ-proof-3} follows since removing conditions does not decrease the conditional entropy;
and \eqref{eq:WZ-proof-5} follows from \eqref{eq:three-variable-identity} applied for $A=(\tilde{U}_J,J)$, $B=\tilde{X}_J$, and $C=\tilde{Y}_J$.

Next, we single-letterize \eqref{eq:WZ-multiletter-bound-1} as follows:
\begin{align}
\lefteqn{ D(P_{\tilde{X}^n \tilde{Y}^n} \| P_{X^n Y^n}) + I(\tilde{M} \wedge \tilde{Y}^n | \tilde{X}^n) + I(\tilde{Z}^n \wedge \tilde{X}^n|\tilde{M},\tilde{Y}^n) } \\ 
&= D(P_{\tilde{X}^n \tilde{Y}^n} \| P_{X^n Y^n}) + H(\tilde{Y}^n|\tilde{X}^n) - H(\tilde{Y}^n|\tilde{M},\tilde{X}^n) + H(\tilde{X}^n|\tilde{M},\tilde{Y}^n) 
 - H(\tilde{X}^n| \tilde{M}, \tilde{Y}^n,\tilde{Z}^n)
 \label{eq:WZ-proof-6} \\
&=  D(P_{\tilde{X}^n \tilde{Y}^n} \| P_{X^n Y^n}) + H(\tilde{Y}^n|\tilde{X}^n)  
 - H(\tilde{X}^n| \tilde{M}, \tilde{Y}^n,\tilde{Z}^n) + H(\tilde{X}^n|\tilde{M}) - H(\tilde{Y}^n|\tilde{M})
 \label{eq:WZ-proof-6b}  \\
&\ge n \big[ D( P_{\tilde{X}_J \tilde{Y}_J} \| P_{XY}) + H(\tilde{Y}_J|\tilde{X}_J) - H(\tilde{X}_J | \tilde{U}_J,J,\tilde{Y}_J,\tilde{Z}_J) \big] 
  + H(\tilde{X}^n|\tilde{M}) - H(\tilde{Y}^n|\tilde{M}) \label{eq:WZ-proof-7} \\
&=   n \big[ D( P_{\tilde{X}_J \tilde{Y}_J} \| P_{XY}) + H(\tilde{Y}_J|\tilde{X}_J) - H(\tilde{X}_J | \tilde{U}_J,J,\tilde{Y}_J,\tilde{Z}_J) 
  + H(\tilde{X}_J|\tilde{U}_J,J) - H(\tilde{Y}_J|\tilde{U}_J,J)  \big] \label{eq:WZ-proof-8} \\
&=  n \big[ D( P_{\tilde{X}_J \tilde{Y}_J} \| P_{XY}) + H(\tilde{Y}_J|\tilde{X}_J) - H(\tilde{Y}_J|\tilde{U}_J,J,\tilde{X}_J) + H(\tilde{X}_J|\tilde{U}_J,J, \tilde{Y}_J)
 - H(\tilde{X}_J | \tilde{U}_J,J,\tilde{Y}_J,\tilde{Z}_J) 
\big] \label{eq:WZ-proof-9} \\
&= n \big[ D( P_{\tilde{X}_J \tilde{Y}_J} \| P_{XY}) + I(\tilde{U}_J,J \wedge \tilde{Y}_J | \tilde{X}_J) +  I(\tilde{Z}_J \wedge \tilde{X}_J| \tilde{U}_J, J, \tilde{Y}_J)  \big], \label{eq:WZ-proof-10}   
\end{align}  
where 
\eqref{eq:WZ-proof-7} follows from \eqref{eq:single-letterize-divergence-entropy} (where the role of $\tilde{X}^n$ and $\tilde{Y}^n$ interchanged)
and the same manipulations as \eqref{eq:WZ-proof-1}-\eqref{eq:WZ-proof-4} for  $ - H(\tilde{X}^n| \tilde{M}, \tilde{Y}^n,\tilde{Z}^n) $;
\eqref{eq:WZ-proof-8} follows from the sum identity (eg.~see \cite[p.~314]{csiszar-korner:11} or \cite[p.~25]{elgamal-kim-book}):
\begin{align}
H(\tilde{X}^n|\tilde{M}) - H(\tilde{Y}^n|\tilde{M}) = \sum_{j=1}^n\big[ H(\tilde{X}_j | \tilde{M},\tilde{X}_j^-,\tilde{Y}_j^+) - H(\tilde{Y}_j|\tilde{M},\tilde{X}_j^-,\tilde{Y}_j^+) \big].
\end{align}

Consequently, by combining \eqref{eq:WZ-multiletter-bound-1} and  \eqref{eq:WZ-multiletter-bound-2} together with
the single-letter manipulations above, we have
\begin{align*}
\frac{1}{n} \log \frac{1}{\Pc(\Phi_n)} &\ge
 D( P_{\tilde{X}_J \tilde{Y}_J} \| P_{XY}) + I(\tilde{U}_J, J \wedge \tilde{Y}_J | \tilde{X}_J) + I(\tilde{Z}_J \wedge \tilde{X}_J |\tilde{U}_J,J, \tilde{Y}_J) \\
 &\hspace{50mm} + | I( \tilde{U}_J, J \wedge \tilde{X}_J) - I(\tilde{U}_J,J \wedge \tilde{Y}_J) - R|^+.
\end{align*}
This bound together with the distortion constraint \eqref{eq:WZ-distortion-bound} and the cardinality bound to be proved in 
Appendix \ref{appendix:cardinality} imply
(see also the identity \eqref{eq:WZ-soft-Markov})
\begin{align*}
\frac{1}{n} \log \frac{1}{\Pc(\Phi_n)} \ge F^\star(R,D)
\end{align*}
for any code $\Phi_n$, which completes the converse part of Theorem \ref{theorem:WZ}.
\qed

\begin{remark} \label{remark:comparison-to-TyaWat}
Using the same notation as the above converse proof, we can rewrite the single-letterization result in
\cite[Theorem 4]{TyaWat20} as follows:\footnote{We omitted the
distortion constraint term that can be easily handled separately.} for any $\alpha >0$,
\begin{align*}
\lefteqn{
(\alpha+1) D(P_{\tilde{X}^n \tilde{Y}^n} \| P_{X^n Y^n}) + I(\tilde{M}\wedge \tilde{X}^n|\tilde{Y}^n)
+ \alpha\big( I(\tilde{M}\wedge \tilde{Y}^n|\tilde{X}^n) + I(\tilde{Z}^n \wedge \tilde{X}^n|\tilde{M},\tilde{Y}^n) \big)
} \\
&\ge n\big[
(\alpha+1) D(P_{\tilde{X}_J \tilde{Y}_J} \| P_{XY}) + I(\tilde{U}_J,J \wedge \tilde{X}_J | \tilde{Y}_J) 
+ \alpha \big(
I(\tilde{U}_J, J \wedge \tilde{Y}_J | \tilde{X}_J) + I(\tilde{Z}_J \wedge \tilde{X}_J | \tilde{U}_J, J, \tilde{Y}_J)
\big)
\big].
\end{align*}
In the course of this single-letterization, since we need to bound $H(\tilde{X}^n|\tilde{Y}^n)$ and $\alpha H(\tilde{Y}^n|\tilde{X}^n)$
from below using \eqref{eq:single-letterize-divergence-entropy}, we need the factor $(\alpha+1)$ of the divergence term 
$D(P_{\tilde{X}^n\tilde{Y}^n}\| P_{X^nY^n})$, which prevents us to derive the tight exponent. 
A technical novelty of the converse proof in this paper is that
we derive the two single-letterizations, one for $I(\tilde{M}\wedge \tilde{X}^n|\tilde{Y}^n)$ and the other for
$I(\tilde{M} \wedge \tilde{Y}^n|\tilde{X}^n)$ separately, and couple them together via the identity \eqref{eq:three-variable-identity}.

As we mentioned in Section \ref{section:problem-formulation}, the rate-distortion function can be written as
the mutual information difference form \eqref{eq:WZ-RD-difference-form} or the conditional mutual information form \eqref{eq:WZ-RD-conditional-form};
in \cite{TyaWat20}, the latter form was employed,\footnote{The conditional mutual information form was employed in \cite{Oohama:18} as well.} 
and the soft Markov constraint was added on top of that.
As we discussed in the paragraph containing \eqref{eq:identity-two-forms}, the conditional mutual information form can be 
interpreted as the mutual information difference form plus the soft Markov constraint; in other words, the conditional mutual
information form already includes the soft Markov constraint.
It is this insight that enables us to derive the tight strong converse exponent. 
\end{remark}

\section{Proof of Achievability}

In this section, we prove the achievability part of Theorem \ref{theorem:WZ}.
First, for a given $D>0$ and an arbitrarily small positive number $\varepsilon>0$,
we fix an arbitrary joint distribution $P_{\tilde{U}\tilde{X}\tilde{Y}\tilde{Z}}$ satisfying $\mathbb{E}[d(\tilde{X},\tilde{Y},\tilde{Z})]\le D-\varepsilon$.
Let $P_{\bar{U}\bar{X}\bar{Y}\bar{Z}}$ be a joint type satisfying
\begin{align} \label{eq:type-approximation}
d_{\mathtt{var}}(P_{\bar{U}\bar{X}\bar{Y}\bar{Z}}, P_{\tilde{U}\tilde{X}\tilde{Y}\tilde{Z}}) \le \frac{|{\cal U}||{\cal X}||{\cal Y}||{\cal Z}|}{n}.
\end{align} 
Note that, for sufficiently large $n$, the joint type $P_{\bar{U}\bar{X}\bar{Y}\bar{Z}}$ satisfies 
\begin{align} \label{eq:proof-achievability-WZ-type-distortion}
\mathbb{E}[ d(\bar{X},\bar{Y},\bar{Z})] \le \mathbb{E}[d(\tilde{X},\tilde{Y},\tilde{Z})] + \varepsilon \le D.
\end{align}
We consider stochastic encoder and decoder using the exponential matching lemma reviewed in Appendix \ref{appendix:exponential-matching}.
Since we are interested in a lower bound on the non-excess distortion probability, it suffices to construct a code such that
sequences in a fixed type class ${\cal T}_{\bar{X}\bar{Y}}$ can be decoded within the required distortion level $D$. 
For this reason, we only specify code construction for random variable $(X^n,Y^n)$ uniformly distributed on 
${\cal T}_{\bar{X}\bar{Y}}$.\footnote{By an abuse of notation, even though $(X^n,Y^n)$ are different from i.i.d. random variables
used in Section \ref{section:problem-formulation}, we use $(X^n,Y^n)$ here to avoid complicated notations.} 
More precisely, we only specify the behavior of encoder for sequences in the marginal type class ${\cal T}_{\bar{X}}$
and the behavior of decoder for sequences in the marginal type class ${\cal T}_{\bar{Y}}$. For other sequences,
the behaviors of encoder and decoder are defined arbitrarily, but it does not affect our lower bound on the non-excess distortion probability.
 
Let $P_{U^n|X^n}$ be the test channel defined by
\begin{align}
P_{U^n|X^n}(u^n|x^n) &:= \frac{1}{|{\cal T}_{\bar{U}|\bar{X}}(x^n)|} \\
&\le (n+1)^{|{\cal U}||{\cal X}|} \exp\{ - n H(\bar{U}|\bar{X}) \} \label{eq:type-bound-UX}
\end{align}
for $x^n \in {\cal T}_{\bar{X}}$ and $u^n \in {\cal T}_{\bar{U}|\bar{X}}(x^n)$, where the inequality \eqref{eq:type-bound-UX}
follows from the basic property of the conditional type class \cite[Lemma 2.5]{csiszar-korner:11};
and let $Q_{U^n|Y^n}$ be a channel defined by
\begin{align}
Q_{U^n|Y^n}(u^n|y^n) &:= \frac{1}{|{\cal T}_{\bar{U}|\bar{Y}}(y^n)|} \label{eq:type-definition-UY} \\
&\ge \exp\{ - n H(\bar{U}|\bar{Y})\} \label{eq:type-bound-UY}
\end{align}
for $y^n \in {\cal T}_{\bar{Y}}$ and $u^n \in {\cal T}_{\bar{U}|\bar{Y}}(y^n)$, 
where the inequality \eqref{eq:type-bound-UY} again follows from the basic property of the conditional type class \cite[Lemma 2.5]{csiszar-korner:11}.
It should be noted that, in general, the channel
$Q_{U^n|Y^n}$ defined above is different from the induced channel
\begin{align*}
P_{U^n|Y^n}(u^n|y^n) = \sum_{x^n \in {\cal T}_{\bar{X}}} P_{U^n|X^n}(u^n|x^n) P_{X^n|Y^n}(x^n|y^n) 
= \sum_{x^n \in {\cal T}_{\bar{X}}} \frac{1}{|{\cal T}_{\bar{U}|\bar{X}}(x^n)|} \frac{1}{|{\cal T}_{\bar{X}|\bar{Y}}(y^n)|}.
\end{align*} 

Furthermore, let $Q_{Z^n|U^nY^n}$ be the test channel given by
\begin{align}
Q_{Z^n|U^nY^n}(z^n|u^n,y^n) &:= \frac{1}{|{\cal T}_{\bar{Z}|\bar{U}\bar{Y}}(u^n,y^n)|} \label{eq:type-definition-ZUY} \\
&\ge \exp\{- n H(\bar{Z}|\bar{U},\bar{Y})\} \label{eq:type-bound-ZUY}
\end{align}
for $(u^n,y^n) \in {\cal T}_{\bar{U}\bar{Y}}$ and $z^n \in {\cal T}_{\bar{Z}|\bar{U}\bar{Y}}(u^n,y^n)$,
where the inequality \eqref{eq:type-bound-ZUY} is from \cite[Lemma 2.5]{csiszar-korner:11}.

We use the conditional exponential matching lemma for the following parameters:
\begin{align*}
{\cal A} &\leftarrow {\cal T}_{\bar{X}}, \\
{\cal B} &\leftarrow {\cal M} \times {\cal T}_{\bar{Y}}, \\
{\cal C} &\leftarrow {\cal T}_{\bar{U}} \times {\cal M}, \\
P_{C|A}(c|a) &\leftarrow P_{\check{U}^n \check{M}|X^n}(\check{u}^n,\check{m}|x^n) := P_{U^n|X^n}(\check{u}^n|x^n) \frac{1}{|{\cal M}|}, \\
Q_{C|B}(c|b) &\leftarrow Q_{\check{U}^n\check{M}|MY^n}(\check{u}^n,\check{m}|m,y^n) := Q_{U^n|Y^n}(\check{u}^n|y^n) \bol{1}[\check{m}=m], \\
P_{B|AC}(b|a,c) &\leftarrow P_{MY^n|X^n\check{M}}(m,y^n|x^n,\check{m}) := P_{Y^n|X^n}(y^n|x^n) \bol{1}[m=\check{m}],
\end{align*}
where $P_{Y^n|X^n}(\cdot |x^n)$ is the uniform distribution on ${\cal T}_{\bar{Y}|\bar{X}}(x^n)$ for $x^n \in {\cal T}_{\bar{X}}$.
For $X^n$ uniformly distributed on ${\cal T}_{\bar{X}}$, let 
\begin{align*}
(\check{U}^n,\check{M}) := (\check{U}^n,\check{M})_{P_{\check{U}^n \check{M}|X^n}(\cdot,\cdot|X^n) }
\end{align*}
be the sample by the encoder, and the encoder transmit $\check{M}$ to the receiver. 
For $M=\check{M}$ and $Y^n \sim P_{Y^n|X^n}(\cdot |X^n)$, let 
\begin{align*}
(\hat{U}^n,\hat{M}) := (\check{U}^n,\check{M})_{Q_{\check{U}^n\check{M}|MY^n}(\cdot,\cdot|M,Y^n)}
\end{align*}
be the sample by the decoder. Then, the decoder outputs a random sample $Z^n$ generated by $Q_{Z^n|U^nY^n}(\cdot|\hat{U}^n,Y^n)$.

By the conditional exponential matching lemma, conditioned on $(\check{U}^n,X^n,Y^n) \in {\cal T}_{\bar{U}\bar{X}\bar{Y}}$,
we have\footnote{When $(\check{U}^n,X^n,Y^n) \in {\cal T}_{\bar{U}\bar{X}\bar{Y}}$, note that $\check{U}^n \in {\cal T}_{\bar{U}|\bar{Y}}(Y^n)$.
If $\check{U}^n \notin {\cal T}_{\bar{U}|\bar{Y}}(Y^n)$, then $Q_{U^n|Y^n}(\check{U}^n|Y^n)=0$
and \eqref{eq:Poisson-matching-step-2} is interpreted as $1$; and the inequality in \eqref{eq:Poisson-matching-step-3}
does not hold.}
\begin{align}
\Pr\big( \hat{U}^n \neq \check{U}^n \mid X^n, \check{U}^n, \check{M}, M, Y^n \big)
&\le 1 - \bigg(1 + \frac{P_{\check{U}^n \check{M}|X^n}(\check{U}^n,\check{M}|X^n)}{Q_{\check{U}^n\check{M}|MY^n}(\check{U}^n,\check{M}|M,Y^n)} \bigg)^{-1} \\
&= 1 - \bigg( 1 + \frac{\frac{1}{|{\cal M}| } P_{U^n|X^n}(\check{U}^n|X^n) }{Q_{U^n|Y^n}(\check{U}^n|Y^n)} \bigg)^{-1} \label{eq:Poisson-matching-step-2} \\
&\le 1 - \bigg(1 + (n+1)^{|{\cal U}||{\cal X}|} \exp\big\{ n ( H(\bar{U}|\bar{Y}) - H(\bar{U}|\bar{X}) -R) \big\} \bigg)^{-1} \label{eq:Poisson-matching-step-3} \\
&= 1 - \bigg(1 + (n+1)^{|{\cal U}||{\cal X}|} \exp\big\{ n (I(\bar{U} \wedge \bar{X}) - I(\bar{U} \wedge \bar{Y})-R) \big\} \bigg)^{-1},
\end{align}
where the second inequality follows from \eqref{eq:type-bound-UX} and \eqref{eq:type-bound-UY}.
Furthermore, since 
\begin{align*}
\lefteqn{
1 + (n+1)^{|{\cal U}||{\cal X}|} \exp\big\{ n (I(\bar{U} \wedge \bar{X}) - I(\bar{U} \wedge \bar{Y})-R) \big\} } \\
&\le 2 (n+1)^{|{\cal U}||{\cal X}|} \exp\big\{ n |I(\bar{U} \wedge \bar{X}) - I(\bar{U} \wedge \bar{Y})-R|^+ \big\}, 
\end{align*}
conditioned on $(\check{U}^n,X^n,Y^n) \in {\cal T}_{\bar{U}\bar{X}\bar{Y}}$, we have
\begin{align*}
\lefteqn{ \Pr\big( \hat{U}^n \neq \check{U}^n \mid X^n, \check{U}^n, \check{M}, M, Y^n \big) } \\
&\le 1 - \frac{(n+1)^{-|{\cal U}||{\cal X}|}}{2} \exp\big\{ - n |I(\bar{U} \wedge \bar{X}) - I(\bar{U} \wedge \bar{Y})-R|^+ \big\}.
\end{align*}

Since 
\begin{align} \label{eq:conditional-type-class-ZUXY}
|{\cal T}_{\bar{Z}|\bar{U}\bar{X}\bar{Y}}(u^n,x^n,y^n)| \ge (n+1)^{-|{\cal U}||{\cal X}||{\cal Y}||{\cal Z}|} \exp\{ n H(\bar{Z}|\bar{U},\bar{X},\bar{Y})\}
\end{align}
for any $(u^n,x^n,y^n) \in {\cal T}_{\bar{U}\bar{X}\bar{Y}}$ \cite[Lemma 2.5]{csiszar-korner:11}, by combining with \eqref{eq:type-bound-ZUY}, we have
\begin{align}
\Pr\big( Z^n \in {\cal T}_{\bar{Z}|\bar{U}\bar{X}\bar{Y}}(\check{U}^n,X^n,Y^n) \mid X^n,Y^n, \check{U}^n,\hat{U}^n, \hat{U}^n = \check{U}^n \big)
\ge (n+1)^{-|{\cal U}||{\cal X}||{\cal Y}||{\cal Z}|} \exp\{ - n I(\bar{Z} \wedge \bar{X} | \bar{U},\bar{Y})\}.
\label{eq:lower-probability-conditional-type-class-ZUXY}
\end{align}
Similarly, since
\begin{align*}
|{\cal T}_{\bar{U}|\bar{X}\bar{Y}}(x^n,y^n)| \ge (n+1)^{-|{\cal U}||{\cal X}||{\cal Y}|} \exp\{ n H(\bar{U}|\bar{X},\bar{Y})\}
\end{align*}
for any $(x^n,y^n) \in {\cal T}_{\bar{X}\bar{Y}}$ and
\begin{align*}
P_{U^n|X^n}(u^n|x^n) \ge \exp\{ - n H(\bar{U}|\bar{X})\},
\end{align*}
by noting that the exponential matching guarantees $\check{U}^n \sim P_{U^n|X^n}(\cdot|X^n)$,
we have
\begin{align}
\mathbb{E}\big[ \bol{1}[ (\check{U}^n,X^n,Y^n) \in {\cal T}_{\bar{U}\bar{X}\bar{Y}} \big]
&= \Pr\big( (\check{U}^n,X^n,Y^n) \in {\cal T}_{\bar{U}\bar{X}\bar{Y}} \big) \\
&\ge (n+1)^{-|{\cal U}||{\cal X}||{\cal Y}|} \exp\{ - n I(\bar{U}\wedge \bar{Y}|\bar{X})\}.
\label{eq:probability-lower-bound-UXY}
\end{align}
By combining these estimates, we have
\begin{align*}
\lefteqn{ \Pr\big( (\check{U}^n, X^n,Y^n,Z^n) \in {\cal T}_{\bar{U}\bar{X}\bar{Y}\bar{Z}} \big) } \\
&\ge \mathbb{E}\bigg[ \bol{1}\big[ (\check{U}^n,X^n,Y^n) \in {\cal T}_{\bar{U}\bar{X}\bar{Y}} \big] \cdot
\Pr\big( \hat{U}^n = \check{U}^n \mid X^n, \check{U}^n, \check{M}, M, Y^n \big)  \\
&\hspace{60mm} \cdot \Pr\big( Z^n \in {\cal T}_{\bar{Z}|\bar{U}\bar{X}\bar{Y}}(\check{U}^n,X^n,Y^n) \mid X^n,Y^n, \check{U}^n,\hat{U}^n, \hat{U}^n = \check{U}^n \big)
  \bigg] \\
&\ge \frac{(n+1)^{-|{\cal U}||{\cal X}|(|{\cal Y}||{\cal Z}|+1)}}{2} 
\mathbb{E}\big[ \bol{1}[ (\check{U}^n,X^n,Y^n) \in {\cal T}_{\bar{U}\bar{X}\bar{Y}}] \big] \\
&\hspace{60mm} \cdot \exp\big\{ - n ( I(\bar{Z}\wedge \bar{X} | \bar{U},\bar{Y}) + |I(\bar{U} \wedge \bar{X}) - I(\bar{U} \wedge \bar{Y})-R|^+ ) \big\} \\ 
&\ge \frac{(n+1)^{-|{\cal U}||{\cal X}|(|{\cal Y}|(|{\cal Z}|+1)+1)}}{2}
  \exp\big\{ - n ( I(\bar{U}\wedge \bar{Y}|\bar{X})+I(\bar{Z}\wedge \bar{X} | \bar{U},\bar{Y}) + |I(\bar{U} \wedge \bar{X}) - I(\bar{U} \wedge \bar{Y})-R|^+ ) \big\}.  
\end{align*}
By \eqref{eq:proof-achievability-WZ-type-distortion}, we have
\begin{align*}
\lefteqn{ \Pr\big( d_n(X^n,Y^n,Z^n) \le D \big) } \\
 &\ge \Pr\big( (\check{U}^n, X^n,Y^n,Z^n) \in {\cal T}_{\bar{U}\bar{X}\bar{Y}\bar{Z}} \big) \\
&\ge \frac{(n+1)^{-|{\cal U}||{\cal X}|(|{\cal Y}|(|{\cal Z}|+1)+1)}}{2}
 \exp\big\{ - n ( I(\bar{U}\wedge \bar{Y}|\bar{X})+I(\bar{Z}\wedge \bar{X} | \bar{U},\bar{Y}) + |I(\bar{U} \wedge \bar{X}) - I(\bar{U} \wedge \bar{Y})-R|^+ ) \big\}.  
\end{align*}
Finally, by noting that the type class ${\cal T}_{\bar{X}\bar{Y}}$ occur with
\begin{align} \label{eq:probability-lower-bound-XY}
P_{XY}^n({\cal T}_{\bar{X}\bar{Y}}) \ge (n+1)^{-|{\cal X}||{\cal Y}|}\exp\big\{ - n D(P_{\bar{X}\bar{Y}}\|P_{XY})\},
\end{align}
we have a code $\Phi_n$ such that
\begin{align*}
\lefteqn{ \Pc(\Phi_n) \ge \frac{(n+1)^{-|{\cal X}|(|{\cal U}||{\cal Y}|(|{\cal Z}|+1)+|{\cal U}|+|{\cal Y}|)}}{2} } \\ 
& \hspace{20mm} \cdot  
\exp\big\{ - n (D(P_{\bar{X}\bar{Y}}\|P_{XY}) + I(\bar{U}\wedge \bar{Y}|\bar{X})+I(\bar{Z}\wedge \bar{X} | \bar{U},\bar{Y}) 
 + |I(\bar{U} \wedge \bar{X}) - I(\bar{U} \wedge \bar{Y})-R|^+) \big\}.
\end{align*}
Thus, for sufficiently large $n$, by the continuity of information quantities and \eqref{eq:type-approximation}, we have
\begin{align} \label{eq:achievability-general}
F(R,D) \le D(P_{\tilde{X}\tilde{Y}}\|P_{XY}) + I(\tilde{U}\wedge \tilde{Y}|\tilde{X}) 
+ I(\tilde{Z}\wedge \tilde{X}|\tilde{U},\tilde{Y}) + |I(\tilde{U}\wedge \tilde{X})-I(\tilde{U}\wedge \tilde{Y})-R|^+ + \varepsilon.
\end{align}
Since $P_{\tilde{U}\tilde{X}\tilde{Y}\tilde{Z}}$ can be arbitrary as long as $\mathbb{E}[d(\tilde{X},\tilde{Y},\tilde{Z})]\le D-\varepsilon$, we have
\begin{align*}
F(R,D) \le F^\star(R,D-\varepsilon) + \varepsilon.
\end{align*}
Finally, since $\varepsilon>0$ can be arbitrary, we have the desired achievability bound. \qed

\begin{remark} \label{remark:proof-dfc}
Since the above proof cannot be applied for $D=0$, we need to consider a separate proof 
for the distributed function computation case (cf.~Corollary \ref{corollary:dfc}).
For an arbitrarily fixed joint distribution $P_{\tilde{U}\tilde{X}\tilde{Y}}$,
let $P_{\bar{U}\bar{X}\bar{Y}}$ be an approximation by type as in \eqref{eq:type-approximation}, and let $\bar{Z}=f(\bar{X},\bar{Y})$. 
Code construction and analyses are almost the same as above 
except the following. The cardinality of the type class ${\cal T}_{\bar{Z}|\bar{U}\bar{X}\bar{Y}}(u^n,x^n,y^n)$ in \eqref{eq:conditional-type-class-ZUXY}
is $1$; thus, the lower bound in \eqref{eq:lower-probability-conditional-type-class-ZUXY} becomes
$\exp\{ - n H(\bar{Z}|\bar{U},\bar{Y})\}$. Finally, for $f_n(x^n,y^n)=(f(x_1,y_1),\ldots,f(x_n,y_n))$, we have
\begin{align*}
\Pr\big( Z^n = f_n(X^n,Y^n) \big) \ge \Pr\big( (\check{U}^n,X^n,Y^n,Z^n) \in {\cal T}_{\bar{U}\bar{X}\bar{Y}\bar{Z}} \big),
\end{align*}
and the right hand side can be bounded in the same manner as above. 
\end{remark}

\begin{remark} \label{ramark:naive-achievability}
By a much naive argument, for an arbitrarily small $\varepsilon>0$,
we can show that
\begin{align} \label{eq:achievability-special}
F(R,D) \le D(P_{\tilde{X}\tilde{Y}}\|P_{XY}) + I(\tilde{U}\wedge \tilde{X}|\tilde{Y}) + I(\tilde{Z}\wedge \tilde{X}|\tilde{U},\tilde{Y}) + \varepsilon
\end{align}
for an arbitrarily fixed $P_{\tilde{U}\tilde{X}\tilde{Y}\tilde{Z}}$ satisfying $\mathbb{E}[d(\tilde{X},\tilde{Y},\tilde{Z})]\le D-\varepsilon$.
Since $\varepsilon>0$ can be arbitrary, we have
\begin{align} \label{eq:naive-achievability}
F(R,D) \le \min_{P_{\tilde{U}\tilde{X}\tilde{Y}\tilde{Z}} \atop \mathbb{E}[d(\tilde{X},\tilde{Y},\tilde{Z})]\le D}
\big[
D(P_{\tilde{X}\tilde{Y}}\| P_{XY}) + I(\tilde{U} \wedge \tilde{X}|\tilde{Y}) + I(\tilde{Z}\wedge \tilde{X}|\tilde{U},\tilde{Y})
\big].
\end{align}
In fact, when $I(\tilde{U}\wedge \tilde{X})- I(\tilde{U}\wedge \tilde{Y}) < 0$, the achievability bound \eqref{eq:achievability-special} is tighter than
\eqref{eq:achievability-general}. However, as discussed in Remark \ref{remark:optimizer}, such $P_{\tilde{U}\tilde{X}\tilde{Y}\tilde{Z}}$ is not an optimizer.

To verify \eqref{eq:achievability-special}, let $P_{\bar{U}\bar{X}\bar{Y}\bar{Z}}$ be approximation by type as in \eqref{eq:type-approximation}.
Let $(X^n,Y^n)$ be uniformly distributed on ${\cal T}_{\bar{X}\bar{Y}}$, and let $Q_{U^n|Y^n}$ and $Q_{Z^n|U^nY^n}$ be defined 
by \eqref{eq:type-definition-UY} and \eqref{eq:type-definition-ZUY}, respectively. Without using any communication, the decoder 
locally samples $\hat{U}^n$ generated by $Q_{U^n|Y^n}(\cdot|Y^n)$ and $Z^n$ generated by $Q_{Z^n|U^nY^n}(\cdot|\hat{U}^n,Y^n)$.
Then, by a similar reasoning as \eqref{eq:lower-probability-conditional-type-class-ZUXY}, 
\eqref{eq:probability-lower-bound-UXY} used for $Q_{U^n|Y^n}$ instead of $P_{U^n|X^n}$, 
and \eqref{eq:probability-lower-bound-XY}, we have
\begin{align*}
\Pc(\Phi_n) \ge (n+1)^{|{\cal X}||{\cal Y}|(|{\cal U}|(|{\cal Z}|+1)+1)}
\exp\big\{ - n (D(P_{\bar{X}\bar{Y}}\|P_{XY}) + I(\bar{U}\wedge \bar{X}|\bar{Y})+I(\bar{Z}\wedge \bar{X} | \bar{U},\bar{Y}) \big\},
\end{align*}
which implies \eqref{eq:achievability-special} for sufficiently large $n$.
\end{remark}

\begin{remark} \label{remark:deviation-Markovity}
In \cite{TakWat:25}, in order to identify the role of soft Markov constraint in the achievability, we considered 
(in the terminology of the WZ problem) the ``deviation of Markovity in forward direction" in the following sense.\footnote{Here,
the terms ``forward" or ``reverse" are used to indicate whether the auxiliary random variable appears on the conditional side or not;
these terms are introduced to emphasize that the soft Markov constraint appears by different principles in each case.}
Note that we can write (see also \eqref{eq:WZ-soft-Markov})
\begin{align*}
D(P_{\bar{X}\bar{Y}} \| P_{XY}) + I(\bar{U} \wedge \bar{Y}|\bar{X}) = D(P_{\bar{X}}\|P_X) + D(P_{\bar{Y}|\bar{U}\bar{X}}\|P_{Y|X}|P_{\bar{U}\bar{X}}).
\end{align*}
Then, the term $D(P_{\bar{Y}|\bar{U}\bar{X}}\|P_{Y|X}|P_{\bar{U}\bar{X}})$ appears as the exponent of the probability
\begin{align*}
P_{Y|X}^n({\cal T}_{\bar{Y}|\bar{U}\bar{X}}(u^n,x^n)|x^n) \stackrel{\cdot}{=} \exp\{ - n D(P_{\bar{Y}|\bar{U}\bar{X}}\|P_{Y|X}|P_{\bar{U}\bar{X}}) \}
\end{align*}
for $(u^n,x^n) \in {\cal T}_{\bar{U}\bar{X}}$, where $\stackrel{\cdot}{=}$ is the equality up to the polynomial factor of $n$.

On the other hand, in our achievability proof above, we considered the ``deviation of Markovity in reverse direction." More specifically, the soft
Markov constraint $I(\bar{U}\wedge \bar{Y}|\bar{X})$ appears in \eqref{eq:probability-lower-bound-UXY} via the decomposition
\begin{align*}
I(\bar{U}\wedge \bar{Y}|\bar{X}) = H(\bar{U}|\bar{X}) - H(\bar{U}|\bar{X},\bar{Y}).
\end{align*}
The soft Markov constraint $I(\bar{Z}\wedge \bar{X}|\bar{U},\bar{Y})$ also appears 
in \eqref{eq:lower-probability-conditional-type-class-ZUXY} as the deviation of Markovity in reverse direction.
We do not know if it is possible to prove the achievability in such a manner that 
$I(\bar{Z}\wedge \bar{X}|\bar{U},\bar{Y})$ appears as the deviation of Markovity in forward direction.

In the derivation of \eqref{eq:probability-lower-bound-UXY}, 
it is crucial that $\check{U}^n$ is sampled according to $P_{U^n|X^n}(\cdot |X^n)$ exactly.
One advantage of the Poisson matching (exponential matching) approach of \cite{LiAna:21} is that 
sampling of $\check{U}^n \sim P_{U^n|X^n}(\cdot |X^n)$ is exact; in some other approach, such as the channel simulation
via soft covering \cite{cuff:12}, sampling is only approximate and not suitable for our analysis. 
\end{remark}


\appendix

\subsection{Exponential Matching Lemma} \label{appendix:exponential-matching}

In the achievability proof, we use the Poisson matching lemma introduced in \cite{LiAna:21}.
Since we only consider discrete alphabets in this paper, we review a simplified version of this approach, the exponential matching lemma, 
in this appendix. 

For a given probability distribution $P \in {\cal P}({\cal C})$ on a finite alphabet ${\cal C}$, we consider a problem 
of sampling a random variable distributed according to $P$. Let $Z_1,\ldots,Z_{|{\cal C}|} \sim \rom{Exp}(1)$ be independent
random variables distributed according to the exponential distribution with rate $1$, 
and let 
\begin{align} \label{eq:sampling-P}
\check{C}_P := \argmin_{c \in {\cal C}: \atop P(c)>0} \frac{Z_c}{P(c)}.
\end{align}
From a property of the exponential distribution (eg.~see \cite[Lemma 8.5]{MitzenmacherUpfal:book2}), it can be verified that
$\check{C}_P$ is distributed according to $P$. 

Now, let us consider the same sampling procedure using another distribution $Q \in {\cal P}({\cal C})$, i.e., let 
\begin{align} \label{eq:sampling-Q}
\check{C}_Q := \argmin_{c \in {\cal C}: \atop Q(c)>0} \frac{Z_c}{Q(c)},
\end{align}
where we use the same exponential random variables $Z_1,\ldots,Z_{|{\cal C}|}$ as \eqref{eq:sampling-P}.
Then, how $\check{C}_P$ and $\check{C}_Q$ are related. The exponential matching lemma provides an upper bound on the
probability of disagreement.
\begin{lemma}[Exponential matching lemma \cite{LiAna:21}] \label{lemma:exponential-matching}
For $\tilde{C}_P$ and $\tilde{C}_Q$ sampled according to, respectively, we have
\begin{align*}
\Pr( \check{C}_P \neq \check{C}_Q \mid \check{C}_P = c) \le 1 - \bigg( 1 + \frac{P(c)}{Q(c)} \bigg)^{-1}
\end{align*}
for every $c \in {\cal C}$ such that $P(c) > 0$.\footnote{The right hand side is regarded as $1$ if $Q(c)=0$.}
\end{lemma}

For coding applications, we usually use a conditional version of the exponential matching lemma.
Let $P_{ABC} = P_A P_{B|A}P_{C|AB}$ be a joint distribution on ${\cal A}\times {\cal B}\times {\cal C}$.
As before, let $Z_1,\ldots,Z_{|{\cal C}|} \sim \rom{Exp}(1)$ be independent, and let
\begin{align}
\check{C}_a := \check{C}_{P_{C|A}(\cdot|a)} = \argmin_{ c\in {\cal C}: \atop P_{C|A}(c|a) > 0} \frac{Z_c}{P_{C|A}(c|a)}.
\end{align}
For a conditional distribution $Q_{C|B}$, which may not be the same as $P_{C|B}$, let 
\begin{align}
\hat{C}_b := \check{C}_{Q_{C|B}(\cdot|b)} = \argmin_{ c\in {\cal C}: \atop Q_{C|B}(c|b) > 0} \frac{Z_c}{Q_{C|B}(c|b)}.
\end{align}
Then, by the exponential matching lemma (Lemma \ref{lemma:exponential-matching}) applied for $P \leftarrow P_{C|A}(\cdot|a)$
and $Q \leftarrow Q_{C|B}(\cdot|b)$, we have
\begin{align}
\Pr\big( \check{C}_a \neq \hat{C}_b \mid \check{C}_a = c \big) \le 1 - \bigg( 1 + \frac{P_{C|A}(c|a)}{Q_{C|B}(c|b)} \bigg)^{-1}
\end{align}
for every $a\in {\cal A}$ and $b\in {\cal B}$ and every $c \in {\cal C}$ with $P_{C|A}(c|a)>0$.

For $A \sim P_A$, note that $(A, \check{C}_A) \sim P_{AC}$. Then, let $B \sim P_{B|AC}(\cdot | A, \check{C}_A)$.
Note that $B \mc (A, \check{C}_A) \mc \{ Z_c\}_{c\in {\cal C}}$ holds, i.e., conditioning on $B=b$
does not affect the conditional distribution of $\{ Z_c\}_{c\in {\cal C}}$ given
$A=a$ and $\check{C}_a=c$. Thus, we have the following conditional version of Lemma \ref{lemma:exponential-matching}.
\begin{lemma}[Conditional exponential matching lemma \cite{LiAna:21}] \label{lemma:conditional-exponential-matching}
\begin{align}
\Pr\big( \check{C}_A \neq \hat{C}_B \mid A=a, B=b, \check{C}_A = c\big) \le 1 - \bigg( 1 + \frac{P_{C|A}(c|a)}{Q_{C|B}(c|b)} \bigg)^{-1}.
\end{align}
\end{lemma}

\subsection{Cardinality bounds} \label{appendix:cardinality}

By the support lemma \cite[Lemma 15.4]{csiszar-korner:11}, we can restrict the cardinality of $\tilde{U}$ to 
$|{\cal U}|\le |{\cal X}||{\cal Y}||{\cal Z}|+1$ as follows. We set the following functions on ${\cal P}({\cal X}\times {\cal Y}\times {\cal Z})$:
\begin{align*}
g_1(P_{\tilde{X}\tilde{Y}\tilde{Z}}) &= H(\tilde{Z}|\tilde{Y}) - H(\tilde{Z}|\tilde{X},\tilde{Y}) - H(\tilde{Y}|\tilde{X}), \\
g_2(P_{\tilde{X}\tilde{Y}\tilde{Z}}) &= H(\tilde{Y})-H(\tilde{X}).
\end{align*}
Then, we observe that
\begin{align*}
P_{\tilde{X}\tilde{Y}\tilde{Z}}(x,y,z) &= \sum_u P_{\tilde{U}}(u) P_{\tilde{X}\tilde{Y}\tilde{Z}|\tilde{U}}(x,y,z|u), \\
H(\tilde{Z}|\tilde{U},\tilde{Y}) - H(\tilde{Z}|\tilde{U},\tilde{X},\tilde{Y}) - H(\tilde{Y}|\tilde{U},\tilde{X}) &= \sum_u P_{\tilde{U}}(u) 
 g_1(P_{\tilde{X}\tilde{Y}\tilde{Z}|\tilde{U}}(\cdot,\cdot,\cdot|u)), \\
 H(\tilde{Y}|\tilde{U}) - H(\tilde{X}|\tilde{U}) &=  \sum_u P_{\tilde{U}}(u) 
 g_2(P_{\tilde{X}\tilde{Y}\tilde{Z}|\tilde{U}}(\cdot,\cdot,\cdot|u)).
\end{align*}
Note that
\begin{align*}
I(\tilde{U}\wedge \tilde{Y}|\tilde{X})+I(\tilde{Z}\wedge \tilde{X}|\tilde{U},\tilde{Y})  &= H(\tilde{Y}|\tilde{X}) + H(\tilde{Z}|\tilde{U},\tilde{Y}) - H(\tilde{Z}|\tilde{U},\tilde{X},\tilde{Y}) - H(\tilde{Y}|\tilde{U},\tilde{X}), \\
I(\tilde{U}\wedge \tilde{X}) - I(\tilde{U}\wedge \tilde{Y}) &= H(\tilde{X}) - H(\tilde{Y}) +  H(\tilde{Y}|\tilde{U}) - H(\tilde{X}|\tilde{U}). 
\end{align*}
When $P_{\tilde{X}\tilde{Y}\tilde{Z}}(x,y,z)$ is preserved, $D(P_{\tilde{X}\tilde{Y}}\|P_{XY})$ and $\mathbb{E}[d(\tilde{X},\tilde{Y},\tilde{Z})]$ are preserved. 
Therefore, by the support lemma, it suffices to take $|{\cal U}|\le |{\cal X}||{\cal Y}||{\cal Z}|-1 + 2 = |{\cal X}||{\cal Y}||{\cal Z}|+1$.

\bibliographystyle{./IEEEtranS}
\bibliography{../../../09-04-17-bibtex/reference}

\end{document}